\documentclass[letterpaper,11pt]{article}
\usepackage[utf8]{inputenc}

\usepackage{amsmath, amsthm, amssymb, thm-restate}
\usepackage{algorithmicx}
\usepackage[table,xcdraw]{xcolor}
\usepackage[vlined,linesnumbered]{algorithm2e}
\usepackage{setspace}
\usepackage{mathtools}
\usepackage[numbers]{natbib}
\usepackage{comment} 
\usepackage{tcolorbox} 
\usepackage{xfrac}
\usepackage{hyperref}
\usepackage{multirow}
\usepackage{caption}
\usepackage{bm}
\usepackage{newfloat}
\usepackage{enumitem}
\usepackage{dblfloatfix} 
\usepackage{wrapfig}
\usepackage{dsfont}
\usepackage{mdframed}
\usepackage{multirow}
\usepackage{tikz}

\usepackage{fancyhdr}
\usepackage[noabbrev,nameinlink]{cleveref}

\usepackage[margin=1in]{geometry}

\allowdisplaybreaks

\definecolor{mygreen}{RGB}{10,110,230}
\definecolor{myred}{RGB}{10,110,230}

\hypersetup{
     colorlinks=true,
     citecolor= mygreen,
     linkcolor= myred
}

\renewcommand{\epsilon}{\varepsilon}

\newcommand{\rnk}[0]{\textsc{Ranking}}

\newcommand{\hiddencomment}[1]{}

\newcounter{protocolcounter}
 
\crefname{protocolcounter}{Algorithm}{Algorithms}

\newcommand{\opt}[0]{\text{OPT}}

\newcommand{\wis}[0]{\textsc{WIS}}

\SetKwRepeat{Do}{do}{while}

\renewcommand{\epsilon}[0]{\ensuremath{\varepsilon}}

\let\originalleft\left
\let\originalright\right
\renewcommand{\left}{\mathopen{}\mathclose\bgroup\originalleft}
\renewcommand{\right}{\aftergroup\egroup\originalright}

\crefname{lemma}{Lemma}{Lemmas}
\crefname{theorem}{Theorem}{Theorems}
\crefname{property}{Property}{Properties}
\crefname{claim}{Claim}{Claims}
\crefname{result}{Result}{Results}
\crefname{definition}{Definition}{Definitions}
\crefname{observation}{Observation}{Observations}
\crefname{proposition}{Proposition}{Propositions}
\crefname{assumption}{Assumption}{Assumptions}
\crefname{line}{Line}{Lines}
\crefname{figure}{Figure}{Figures}
\creflabelformat{property}{(#1)#2#3}
\crefname{equation}{}{}
\crefname{section}{Section}{Sections}
\crefname{appendix}{Appendix}{Appendices}
\crefname{algCounter}{Algorithm}{Algorithms}
\Crefname{algCounter}{Algorithm}{Algorithms}

\newtheorem{lemma}{Lemma}[section]
\newtheorem{theorem}[lemma]{Theorem}
\newtheorem{proposition}[lemma]{Proposition}

\newtheorem{definition}[lemma]{Definition}
\newtheorem{claim}[lemma]{Claim}

\newtheorem*{remark*}{Remark}

\makeatletter
\def\thm@space@setup{%
  \thm@preskip= 0.2cm
  \thm@postskip=\thm@preskip 
}
\makeatother

\definecolor{mygreen}{RGB}{20,155,20}
\definecolor{myred}{RGB}{195,20,20}
\definecolor{linkcolor}{RGB}{0,0,230}
\definecolor{mylightgray}{RGB}{230,230,230}
\definecolor{verylightgray}{RGB}{240,240,240}
\definecolor{commentcolor}{RGB}{120,120,120}

\SetCommentSty{mycommfont}

\hypersetup{
     colorlinks=true,
     citecolor= mygreen,
     linkcolor= myred,
     urlcolor= mygreen
}

\newcounter{myalgctr}

\newenvironment{tbox}{
\par\addvspace{0.2cm}
\begin{tcolorbox}[width=\textwidth,
                  boxsep=2pt,
                  left=1pt,
                  right=1pt,
                  top=4pt,
                  boxrule=1pt,
                  arc=0pt,
                  colback=white,
                  colframe=black
                  ]
}{
\end{tcolorbox}
}

\newenvironment{tboxh}{
\par\addvspace{0.2cm}
\begin{tcolorbox}[width=\textwidth,
                  boxsep=2pt,
                  left=1pt,
                  right=1pt,
                  top=4pt,
                  boxrule=1pt,
                  arc=0pt,
                  colback=white,
                  colframe=black,
                  float=t
                  ]
}{
\end{tcolorbox}
}

\newcommand{\tboxhrule}[0]{\vspace{0.1cm} \hrule \vspace{0.2cm}}

\newenvironment{titledtbox}[1]{\begin{tbox}#1 \tboxhrule}{\end{tbox}}
\newenvironment{titledtboxh}[1]{\begin{tboxh}#1 \tboxhrule}{\end{tboxh}}

\makeatletter
\renewcommand{\paragraph}{%
  \@startsection{paragraph}{4}%
  {\z@}{10pt}{-1em}%
  {\normalfont\normalsize\bfseries}%
}
\makeatother

\makeatletter
\patchcmd{\@algocf@start}
  {-1.5em}
  {0pt}
  {}{}
\makeatother

\title{A Simple Analysis of Ranking in General Graphs}

 \author{
 Mahsa Derakhshan \\ {\em Northeastern University} \and 
 Mohammad Roghani  \\ {\em Stanford University} \and
 Mohammad Saneian  \\ {\em Northeastern University} \and
 Tao Yu \\ {\em Northeastern University}
 }
\date{}

\begin{document}
\maketitle

\setcounter{page}{1}

We provide a simple combinatorial analysis of the \textsc{Ranking} algorithm, originally introduced in the seminal work by Karp, Vazirani, and Vazirani [KVV90], demonstrating that it achieves a $(1/2 + c)$-approximate matching for general graphs for  $c \geq  0.005$.

\vspace{30pt}



\vspace{-1cm}
\section{Introduction}
In this work, we study a randomized greedy matching algorithm called {\em \rnk{}} for general graphs. This algorithm was first introduced in the seminal work of Karp, Vazirani, and Vazirani~\cite{karp1990optimal} in 1990 for online bipartite matching with one-sided vertex arrivals and was subsequently extended to general graphs~\cite{DBLP:conf/focs/GoelT12} and bipartite graphs~\cite{DBLP:conf/stoc/MahdianY11} with random vertex arrivals. The algorithm belongs to a class of matching algorithms called vertex-iterative (VI) randomized greedy matching algorithms. These algorithms first draw a permutation $\pi$ over the vertices uniformly at random. They then iterate over the vertices in the order of $\pi$, matching each vertex to one of its available neighbors (if any) according to a specific neighbor selection policy. In \rnk{} for general graphs, this policy selects the first unmatched neighbor in the order given by the same permutation $\pi$ used for the initial iteration.

Besides their simplicity and applicability to settings such as online matching, a significance of VI randomized greedy matching algorithms is that they outperform deterministic greedy algorithms and the randomized edge-iterative greedy algorithm\footnote{The algorithm that draws a permutation $\pi$ over the edges uniformly at random and greedily includes the edges according to $\pi$.}~\cite{DBLP:journals/rsa/DyerF91}, achieving an approximation ratio larger than 0.5. The approximation ratio of vertex-iterative randomized greedy matching algorithms is well-understood for bipartite graphs, with the best approximation ratio being between 0.696~\cite{DBLP:conf/stoc/MahdianY11} and 0.75~\cite{DBLP:conf/focs/GoelT12}. However, despite a long line of work~\cite{DBLP:journals/rsa/AronsonDFS95, poloczek2012randomized, chan2014ranking, DBLP:conf/stoc/TangW020, DBLP:conf/focs/GoelT12,0.546ranking}, our understanding of their approximation ratio for general graphs remains limited, with the best-known lower and upper bounds being 0.5469~\cite{0.546ranking} and 0.75~\cite{DBLP:conf/focs/GoelT12}. Aronson, Dyer, Frieze, and Suen~\cite{DBLP:journals/rsa/AronsonDFS95} were the first to show that VI randomized greedy algorithms surpass 0.5 for general graphs by designing a VI algorithm called the modified randomized greedy (MRG)\footnote{In this algorithm, the neighbor selection policy selects one of the unmatched neighbors uniformly at random.} and proving a lower bound of $0.5+1/400,000$ for its approximation ratio. This approximation ratio was later improved to $0.5+1/256$~\cite{poloczek2012randomized}. Both of these papers involve quite complicated combinatorial analyses. Later, in \cite{0.546ranking} and \cite{DBLP:conf/stoc/TangW020}, authors proved 0.5469 and 0.531 approximation guarantees for \rnk{} and MRG, respectively, using factor-revealing linear programs.

In this work, we take a fresh look at the \rnk{} algorithm for general graphs and provide a significantly simpler combinatorial analysis demonstrating that it achieves an approximation ratio of at least 0.505. Our approach is not only more straightforward and intuitive compared to all previous work, but our approximation guarantee also surpasses existing guarantees obtained without solving factor-revealing linear programs.

\paragraph{Our Techniques.} Our proof idea starts with a few simple, well-known observations. Let $\opt$ be the maximum matching of the input graph, which we can assume is a perfect matching w.l.o.g.~\cite{poloczek2012randomized}, and let $R$ be the output of \rnk{}. If $R$ has an approximation ratio smaller than $1/2+c$ for a constant $c>0$, then $R\oplus \opt$ contains at least $(1/2-3c)n/2$ augmenting paths of length three. Furthermore, since $R$ is a maximal matching, the endpoints of these augmenting paths are unmatched in $R$, and they form an independent set. Following these observations, we define a graph structure called a $k$-{\em wasteful independent set} ($k$-\wis{}) in \Cref{def:ltap}. Given a subset of $2k$ vertices, they form a $k$-\wis{} if they are the $2k$ endpoints of $k$ augmenting paths of length three in $R\oplus \opt$. Our analysis then involves a double counting of these $k$-\wis{} for $k=(1/2-3c)n/2$. We provide combinatorial lower and upper bounds for the expected number of $k$-\wis{} in the output of \rnk{} (respectively in \Cref{lem:kltap-all-perms} and \Cref{lem:upperbound-kwis}). We then show that for $c\leq 0.005$, our bounds reach a contradiction, hence proving that \rnk{} in expectation has an approximation ratio of at least $0.505$.

\paragraph{Further Related Work.} After its introduction in the seminal work of Karp, Vazirani, and Vazirani~\cite{karp1990optimal}, \rnk{} and its extensions have been studied extensively in various settings such as online matching with vertex arrivals \cite{DBLP:conf/sosa/EdenFFS21, DBLP:conf/soda/GoelM08, DBLP:journals/sigact/BirnbaumM08, DBLP:conf/soda/DevanurJK13, DBLP:journals/jacm/HuangKTWZZ20, DBLP:conf/soda/0002PTTWZ19, DBLP:conf/focs/0002T0020,0.546ranking}, oblivious matching~\cite{DBLP:conf/stoc/MahdianY11, DBLP:conf/stoc/TangW020, chan2014ranking}, and stochastic matching with query commits~\cite{DBLP:conf/soda/GamlathKS19, DBLP:conf/soda/Derakhshan023}. See the excellent survey  on online matching~\cite{huang2024online} for a more detailed discussion of related work.

\subsection{Preliminaries}

\paragraph{Notation and Definitions.} Throughout the paper, we denote the input graph as $G = (V, E)$, where $|V| = n$. A maximal matching in $G$ is a matching $M$ such that there is no edge $e \in E \setminus M$ for which $M \cup e$ is also a matching. A maximum matching in $G$ is a matching with the maximum number of edges, denoted by $\mu(G)$. Let \opt{} denote a fixed maximum matching of $G$. As discussed in \cite{DBLP:conf/focs/GoelT12, poloczek2012randomized} (see Corollary 2 of \cite{poloczek2012randomized}), one can assume that $G$ contains a perfect matching when analyzing the ratio of the \rnk{} algorithm.

\paragraph{The \rnk{} Algorithm.} We first draw a permutation $\pi$ over the vertices uniformly at random. Then, we iterate over the vertices in order of $\pi$, and match each vertex to its first unmatched neighbor, again in the order of $\pi$, if any. We use $\pi(v)$ to denote the position of vertex $v$ in the permutation selected by \rnk{}.

\paragraph{Augmenting Paths.} An augmenting path $P$ for a matching $M$ in a graph $G$ is a path that alternates between edges of \opt{} and $M$ such that the path starts and ends with vertices that are unmatched in $M$.  The length of the augmenting path is the number of edges it contains. The presence of an augmenting path indicates that the matching $M$ is not maximum. Furthermore, when matching $M$ is far from being maximum, say a 1/2-approximation, there exist many short augmenting paths.

\begin{proposition}[{\cite[Lemma~1]{KonradMM12}}]\label{prop:aug3-path-bound}
Let $\alpha \geq 0$ and $M$ be a maximal matching of $G$ such that $|M| \leq (1/2 + \alpha) \cdot \mu(G)$. Then, at least $(1/2 - 3\alpha)\cdot \mu(G)$ edges of $M$ are in disjoint, length-three augmenting paths.
\end{proposition}

\paragraph{Combinatorial Tools.} We use the following two well-known bounds on the binomial coefficient.

\begin{proposition}[{\cite[Chapter~3]{marshall1979inequalities}}]\label{prop:bionomila-minimize}
Let $a_1, a_2, \ldots, a_n$ be $n$ positive integers such that $\sum_{i = 1}^n a_i = m$. Then, for a positive integer $x$,  $\sum_{i=1}^n {a_i \choose x}$ is minimized when $\{a_1,\ldots, a_n\} = \{\lfloor m/n \rfloor , \lceil m/n \rceil\}$.
\end{proposition}

\begin{proposition}[{\cite[Chapter~11]{thomas2006elements}}]\label{prop:entropy-approximation}
Let $n$ and $\alpha n$ be two positive integers and $1/n \leq \alpha \leq 1/2$. Then, we have ${n \choose \alpha n} \leq 2^{nH(\alpha)}$ where $H(\alpha) = -\alpha \cdot \log_2(\alpha) - (1-\alpha) \cdot \log_2 (1 - \alpha)$.
\end{proposition}

\section{The Analysis}

Throughout the analysis, we assume that the expected approximation ratio of \rnk{} is at most $1/2 + c$ for some rational $c \geq 0$. Then, we prove that if $c$ is smaller than some constant, it leads to a contradiction. This implies a lower bound on the expected approximation ratio of \rnk{}. Moreover, we assume that, without loss of generality, $cn/2$ is an integer. To understand this, consider a graph where the approximation ratio of \rnk{} is $1/2 + c$ and $cn/2$ is not an integer. For $c=a/b$ with integers $a,b$, by copying the same graph $2b$ times, the approximation guarantee of \rnk{} remains unchanged, but now $cn'/2$ is an integer, where $n'=2bn$ is the number of vertices in the new graph. Our proof relies on counting a structure in the graph, referred to as a $k$-wasteful independent set ($k$-\wis), which we define formally in \Cref{def:ltap}. 
Next, in \Cref{lem:kltap-all-perms}, we demonstrate a lower bound on the expected number of $k$-WIS if the approximation ratio of \rnk{} is at most $1/2 + c$. Further, in \Cref{lem:upperbound-kwis}, we prove an upper bound on the number of $k$-WIS. Finally, in \Cref{thm:final}, we combine these bounds to show a lower bound for the constant $c$.

\begin{definition}[$k$-wasteful independent set ($k$-\wis)]\label{def:ltap}
Given matching $R$ the output of \rnk{}, a subset of $2k$ vertices $I$ is a $k$-wasteful independent set, iff vertices in $I$ are end-points of $k$ length-three augmenting paths in $R\oplus \opt{}$. This also implies that $I$ is an independent set of $G$ since \rnk{} outputs a maximal matching and vertices in $I$ are left unmatched in $R$. Figure~\ref{figure:8934} shows an example of a 5-\wis{}.
\end{definition}

\usetikzlibrary{decorations.pathmorphing}

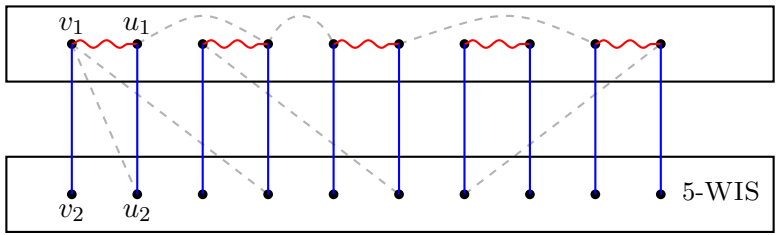
\begin{figure}[h!]
    \centering
    \begin{tikzpicture}[scale=0.5]     
        \pgfmathsetmacro{\scaleFactor}{1.2}

        \definecolor{lightergray}{gray}{0.7}

        \draw[lightergray, thick, dashed] (2.9 * \scaleFactor, 1) .. controls (4.35 * \scaleFactor, 2) and (4.35 * \scaleFactor, 2) .. (5.8 * \scaleFactor, 1);

        \draw[lightergray, thick, dashed] (5.8 * \scaleFactor, 1) .. controls (6.175 * \scaleFactor, 2) and (6.875 * \scaleFactor, 2) .. (7.25 * \scaleFactor, 1);

        \draw[lightergray, thick, dashed] (8.7 * \scaleFactor, 1) .. controls (11.15 * \scaleFactor, 2) and (11.15 * \scaleFactor, 2) .. (13.05 * \scaleFactor, 1);

        \draw[lightergray, thick, dashed] ({1.45 * \scaleFactor}, 1) -- ({2.9 * \scaleFactor}, -3);

        \draw[lightergray, thick, dashed] ({1.45 * \scaleFactor}, 1) -- ({5.8 * \scaleFactor}, -3);

        \draw[lightergray, thick, dashed] ({10.15 * \scaleFactor}, -3) -- ({14.5 * \scaleFactor}, 1);

        \draw[lightergray, thick, dashed] ({4.35 * \scaleFactor}, 1) -- ({8.7 * \scaleFactor}, -3);

        \draw[thick] (0,0) rectangle (17 * \scaleFactor, 2);

        \foreach \i in {1,...,10} {
            \node[circle, fill, inner sep=1.3pt] at ({1.45 * \scaleFactor + 1.45 * \scaleFactor * (\i-1)}, 1) {};
        }

        \node[anchor=south] at ({1.45 * \scaleFactor}, 1) {\(v_1\)};
        \node[anchor=south] at ({2.9 * \scaleFactor}, 1) {\(u_1\)};


        \foreach \i in {0,2,4,6,8} {
            \draw[red, thick, decorate, decoration={snake, amplitude=0.5mm}] ({1.45 * \scaleFactor + 1.45 * \scaleFactor * \i}, 1) -- ({1.45 * \scaleFactor + 1.45 * \scaleFactor * (\i+1)}, 1);
        }

        \draw[thick] (0,-4) rectangle (17 * \scaleFactor, -2);

        \foreach \i in {1,...,10} {
            \node[circle, fill, inner sep=1.3pt] at ({1.45 * \scaleFactor + 1.45 * \scaleFactor * (\i-1)}, -3) {};
        }

        \node[anchor=north] at ({1.45 * \scaleFactor}, -3) {\(v_2\)};
        \node[anchor=north] at ({2.9 * \scaleFactor}, -3) {\(u_2\)};

        \node[anchor=west] at ({17 * \scaleFactor - 2.7}, -2.9) {\small{5-\wis{}}};

        \foreach \i in {1,...,10} {
            \draw[blue, thick] ({1.45 * \scaleFactor + 1.45 * \scaleFactor * (\i-1)}, 1) -- ({1.45 * \scaleFactor + 1.45 * \scaleFactor * (\i-1)}, -3);
        }

    \end{tikzpicture}
   \caption{\small{The red edges are from matching $R$ outputted by \rnk{} and the blue ones are from \opt{}. The lower vertices form a 5-\wis{} as they are the endpoints of five length-three augmenting.}}
    \label{figure:8934}
\end{figure}

\begin{lemma}[Lower Bound on the Expected Number of $k$-\wis{}]\label{lem:kltap-all-perms}
    If the approximation ratio of \rnk{} is at most $1/2 + c$ for $0 \leq c \leq 1/6$, then the expected number of $k$-WIS in the output of \rnk{} is at least one, when $k = (1/2 - 3c) \cdot n/2$.
\end{lemma}
\begin{proof}
    First, if the size of the matching produced by \rnk{} for some permutation is $(1/2 + c) \cdot n/2$, then there exist at least $(1/2 - 3c)\cdot n/2$ disjoint length-three augmenting paths by \Cref{prop:aug3-path-bound}. Let $\mu_1, \ldots, \mu_{n!}$ be the sizes of the matchings produced by \rnk{} for all different permutations such that $\mu_i = (1/2 + c_i)\cdot n/2$ for some $c_i \geq 0$ (such $c_i$ always exists since the output of \rnk{} is a maximal matching). For the $i$th permutation, there exist at least $(1/2 - 3c_i)\cdot n/2$ length-three augmenting paths. Therefore, the total number of length-three augmenting paths in all permutations, denoted by $s$, is at least
    \begin{align*}
        s = \sum_{i=1}^{n!} (\frac{1}{2} - 3c_i)\cdot \frac{n}{2} = \frac{n \cdot n!}{4} - \frac{3n}{2} \sum_{i=1}^{n!} c_i   \geq
        (n!) \cdot \left( \frac{n}{4} - \frac{3nc}{2}\right),
    \end{align*}
    where the last inequality is followed by the fact that the approximation ratio of the algorithm is at least $1/2 + c$. Let $a_i$ be the number of length-three augmenting paths in the $i$th permutation. There exist at least ${ a_i \choose k}$ number of $k$-WIS in the $i$th permutation. It follows that the total number of $k$-WIS is $\sum_{i=1}^{n!} {a_i \choose k}$. Moreover, $\sum_{i=1}^{n!} {a_i \choose k}$ is minimized when all $a_i$'s are equal due to \Cref{prop:bionomila-minimize}. Therefore, the total number of $k$-WIS is at least
    \begin{align*}
        \sum_{i=1}^{n!} {a_i \choose k} \geq n! \cdot {s/n! \choose k} = n! \cdot {n/4 - 3nc/2 \choose n/4 - 3nc/2} = n!,
    \end{align*}
    which implies that the expected number of $k$-WIS is at least one.
\end{proof}

\begin{claim}\label{clm:kltap-bound}
    There exist at most ${n/2 \choose 2k} \cdot 3^k$ \underline{different} $k$-WIS for all possible outputs of \rnk{} on input graph $G$.
\end{claim}
\begin{proof}
    First, we select $2k$ edges from \opt{} that will create $k$ length-three augmenting paths in \Cref{def:ltap}. This selection has ${n/2 \choose 2k}$ possible combinations. To determine which endpoints of these edges will form $I$ (the $k-$\wis{}) in \Cref{def:ltap}, there are $2^{2k}$ possible choices. We further show that this is at most $3^k$. Consider the collection of $k$ length-three alternating paths denoted by $L$. Suppose there is an augmenting path $(v_2, v_1, u_1, u_2)$ in $L$ (for instance, see \Cref{figure:8934}). In this case, we can rule out the possibility of both $v_1$ and $u_1$ being present in $I$ simultaneously, since there is an edge between $v_1$ and $u_1$, and $I$ must be an independent set. Now apply this process to all $k$ augmenting paths in $L$. For each pair of edges in \opt{} that appear in the same length-three augmenting path in $L$, there are 3 possible choices for which endpoint can be included in $I$, since the endpoints of the middle edge in the augmenting path cannot both be in $I$. Consequently, the total number of different possible $k$-WIS is at most ${n/2 \choose 2k} \cdot 3^k$.
\end{proof}

\begin{claim}\label{lem:prob-bound}
Let $I = \{v_1, v_2, \ldots, v_{2k}\}$ be a set of independent vertices in $G$. Then, the vertices in $I$ form a $k$-WIS with probability at most $1/2^{2k}$ in the output of \rnk{}.
\end{claim}

\begin{proof}
Let $E_I$ denote the set of edges in \opt{} that have at least one endpoint in $I$. Since \opt{} is a perfect matching and $I$ is an independent set, we have $|E_I| = 2k$, and $E_I$ covers all vertices in $I$. Let $V(E_I)$ be the set of endpoints of $E_I$ and $\bar{I} = V(E_I) \setminus I$. For the edges in $E_I$ to form length-three augmenting paths, they must be partitioned into pairs, with each pair creating one length-three augmenting path.

Let \((v_1, v_2)\) and \((u_1, u_2)\) be such a pair with \(v_2\) and \(u_2\) in \(I\) (for instance, see \Cref{figure:8934}). Then there must exist an edge between \(v_1\) and \(u_1\) to form the augmenting path. We claim that \(\pi(v_1) < \pi(u_2)\) if $\pi$ is the permutation chosen by \rnk{}. To prove this by contradiction, assume that $\pi(v_1) > \pi(u_2)$. When $u_2$ is processed by \rnk{}, since it is wasted (unmatched in $R$), it must be available and have no free neighbor. In particular, this means $u_1$ is matched and has a match with rank less than $\pi(u_2)$. However, this implies that $\pi(\text{match of } u_1) < \pi(u_2) < \pi(v_1)$, contradicting the fact that $u_1$ is matched with $v_1$ in $R$. Thus, we have \(\pi(v_1) < \pi(u_2)\). We call $u_2$ in this case the {\em counterpart} of $v_1$ and let $C(v_1)$ denote the counterpart of $v_1$. For each vertex in $\bar{I}$, its counterpart is a unique vertex in $I$. Hence, the set of all counterparts of vertices in $\bar{I}$ is $I$. Using the same argument, we can show that for all vertices $v\in \bar{I}, \pi(v)< \pi(C(v))$.

If we consider a permutation of the vertices in $I \cup \bar{I}$, it must satisfy the property described in the above argument for $I$ to be a $k$-\wis{}. In the remainder of the proof, we will focus on calculating the probability of this property being satisfied. Let $\sigma$ be a permutation over vertices in $\bar{I}$. In each iteration, we add a new vertex from $I$ to the permutation such that it does not violate the property. Let $u$ be the last vertex in \(\sigma\). In the permutation of vertices in \(I\) and \(\bar{I}\), \(C(u)\) can appear only after \(u\), and since \(u\) is the last element in \(\bar{I}\), the relative place of \(C(u)\) with respect to \(\sigma\) is only after its last element. For \(C(u)\) to be after all elements in \(\sigma\), which has length \(2k\), it has a probability of \(1/(2k+1)\) since the permutation that \rnk{} chooses is uniformly at random.

Let $u_i$ be the $i$th vertex from the end in $\sigma$. When we add the counterpart of $u_i$, there are $2i - 1$ positions (from the end of the permutation to the one right after $u_i$\footnote{As the counterparts of the last $i-1$ vertices are all placed after $u_i$, there are, in total, $2(i-1)+1=2i-1$ slots where we can put $C(u_i)$ to make $\pi(u_i)<\pi(C(u_i))$ hold.}) among the available $(2k + i)$ positions where we can place it. Therefore, 
\begin{align*}
    \Pr[\text{ $I$ is a $k$-WIS in output of \rnk{} }] \leq& \prod_{i=1}^{2k} \frac{2i - 1}{2k+i}\\
     =& \frac{(4k)! / (\prod_{i=1}^{2k} 2i)}{\prod_{i=1}^{2k} (2k + i)}\\
     =& \frac{(4k)! / ((2k)! \cdot 2^{2k})}{(4k)!/(2k)!} = \frac{1}{2^{2k}}.
\end{align*}
\end{proof}

\begin{lemma}[Upper Bound on the Expected Number of $k$-\wis{}]\label{lem:upperbound-kwis}
    The expected number of $k$-\wis{} in the output of \rnk{} is at most ${n/2 \choose 2k} \cdot (3/4)^k$.
\end{lemma}

\begin{proof}
    By \Cref{clm:kltap-bound}, there exist at most ${n/2 \choose 2k} \cdot 3^k$ different $k$-\wis{}. Also, each of them has a probability $1/(2^{2k})$ to be in the output of \rnk{} by \Cref{lem:prob-bound}. Thus, by the linearity of the expectation, the expected number of $k$-\wis{} in the output of \rnk{} is at most ${n/2 \choose 2k} \cdot (3/4)^k$.
\end{proof}

\begin{theorem}\label{thm:final}
The expected approximation ratio of \rnk{} is at least $0.505$ in general graphs.
\end{theorem}

\begin{proof}
    Suppose that the approximation ratio of \rnk{} is at most $1/2 + c$ for $c \leq 0.005$. Let $k = (1/2 - 3c)\cdot n/2$. By \Cref{lem:kltap-all-perms}, the expected number of $k$-\wis{} is at least one. On the other hand, by \Cref{lem:upperbound-kwis}, the expected number of $k$-\wis{} is at most ${n/2 \choose 2k} \cdot (3/4)^k$. Hence, it must hold that ${n/2 \choose 2k} \cdot (3/4)^k \geq 1$. Now, we show that this cannot happen for $c\leq 0.005$. In particular, we have
        \begin{align*}
        &{n/2 \choose 2k} \cdot \left(\frac{3}{4}\right)^k\\
         =& {n/2 \choose (1-6c) \cdot n/2} \cdot \left(\frac{3}{4}\right)^{(1/2 - 3c)\cdot n/2} & (\text{Since } k = (1/2 - 3c)\cdot n/2)\\
        =&{n/2 \choose 6c \cdot n/2} \cdot \left(\frac{3}{4}\right)^{(1/2 - 3c)\cdot n/2} & (\text{By the identity of binomial coefficients})\\
        \leq& 2^{(n/2) \cdot H(6c)} \cdot \left(\frac{3}{4}\right)^{(1/2 - 3c)\cdot n/2} & (\text{By \Cref{prop:entropy-approximation}})\\
        =&2^{(n/2)\cdot H(6c)}\cdot 2^{(n/2)\cdot \log_2(\frac{3}{4})\cdot(1/2 - 3c)}\\
        =&2^{(n/2)\cdot [H(6c)+\log_2(\frac{3}{4})\cdot(1/2 - 3c)]}\\
        <& 1 &(\text{When $c\leq 0.005,\;H(6c)+\log_2(\frac{3}{4})\cdot(1/2 - 3c)<0$}),
    \end{align*}
    which is a contradiction. Therefore, $c$ cannot be smaller than $0.005$ as otherwise, we obtain ${n/2 \choose 2k} \cdot (3/4)^k < 1$ which completes the proof.
\end{proof}

\bibliographystyle{plain}
\bibliography{references}

@inproceedings{karp1990optimal,
  author       = {Richard M. Karp and
                  Umesh V. Vazirani and
                  Vijay V. Vazirani},
  title        = {An Optimal Algorithm for On-line Bipartite Matching},
  booktitle    = {Proceedings of the 22nd Annual {ACM} Symposium on Theory of Computing,
                  STOC 1990},
  pages        = {352--358},
  publisher    = {{ACM}},
  year         = {1990},
  url          = {https://doi.org/10.1145/100216.100262},
  doi          = {10.1145/100216.100262},
  bibsource    = {dblp computer science bibliography, https://dblp.org}
}

@article{DBLP:journals/rsa/AronsonDFS95,
  author       = {Jonathan Aronson and
                  Martin E. Dyer and
                  Alan M. Frieze and
                  Stephen Suen},
  title        = {Randomized Greedy Matching {II}},
  journal      = {Random Struct. Algorithms},
  volume       = {6},
  number       = {1},
  pages        = {55--74},
  year         = {1995},
  url          = {https://doi.org/10.1002/rsa.3240060107},
  doi          = {10.1002/RSA.3240060107},
  bibsource    = {dblp computer science bibliography, https://dblp.org}
}

@inproceedings{poloczek2012randomized,
  author       = {Matthias Poloczek and
                  Mario Szegedy},
  title        = {Randomized Greedy Algorithms for the Maximum Matching Problem with
                  New Analysis},
  booktitle    = {53rd Annual {IEEE} Symposium on Foundations of Computer Science, {FOCS}
                  2012},
  pages        = {708--717},
  publisher    = {{IEEE} Computer Society},
  year         = {2012},
  url          = {https://doi.org/10.1109/FOCS.2012.20},
  doi          = {10.1109/FOCS.2012.20},
  bibsource    = {dblp computer science bibliography, https://dblp.org}
}

@inproceedings{chan2014ranking,
  author       = {T.{-}H. Hubert Chan and
                  Fei Chen and
                  Xiaowei Wu and
                  Zhichao Zhao},
  title        = {Ranking on Arbitrary Graphs: Rematch via Continuous {LP} with Monotone
                  and Boundary Condition Constraints},
  booktitle    = {Proceedings of the Twenty-Fifth Annual {ACM-SIAM} Symposium on Discrete
                  Algorithms, {SODA} 2014},
  pages        = {1112--1122},
  publisher    = {{SIAM}},
  year         = {2014},
  url          = {https://doi.org/10.1137/1.9781611973402.82},
  doi          = {10.1137/1.9781611973402.82},
  bibsource    = {dblp computer science bibliography, https://dblp.org}
}

@inproceedings{DBLP:conf/focs/GoelT12,
  author       = {Gagan Goel and
                  Pushkar Tripathi},
  title        = {Matching with Our Eyes Closed},
  booktitle    = {53rd Annual {IEEE} Symposium on Foundations of Computer Science, {FOCS}
                  2012},
  pages        = {718--727},
  publisher    = {{IEEE} Computer Society},
  year         = {2012},
  url          = {https://doi.org/10.1109/FOCS.2012.19},
  doi          = {10.1109/FOCS.2012.19},
  bibsource    = {dblp computer science bibliography, https://dblp.org}
}

@inproceedings{DBLP:conf/stoc/TangW020,
  author       = {Zhihao Gavin Tang and
                  Xiaowei Wu and
                  Yuhao Zhang},
  title        = {Towards a better understanding of randomized greedy matching},
  booktitle    = {Proceedings of the 52nd Annual {ACM} {SIGACT} Symposium on Theory
                  of Computing, {STOC} 2020},
  pages        = {1097--1110},
  publisher    = {{ACM}},
  year         = {2020},
  url          = {https://doi.org/10.1145/3357713.3384265},
  doi          = {10.1145/3357713.3384265},
  bibsource    = {dblp computer science bibliography, https://dblp.org}
}

@article{DBLP:journals/rsa/DyerF91,
  author       = {Martin E. Dyer and
                  Alan M. Frieze},
  title        = {Randomized Greedy Matching},
  journal      = {Random Struct. Algorithms},
  volume       = {2},
  number       = {1},
  pages        = {29--46},
  year         = {1991},
  url          = {https://doi.org/10.1002/rsa.3240020104},
  doi          = {10.1002/RSA.3240020104},
  bibsource    = {dblp computer science bibliography, https://dblp.org}
}

@inproceedings{DBLP:conf/soda/DevanurJK13,
  author       = {Nikhil R. Devanur and
                  Kamal Jain and
                  Robert D. Kleinberg},
  title        = {Randomized Primal-Dual analysis of {RANKING} for Online BiPartite
                  Matching},
  booktitle    = {Proceedings of the Twenty-Fourth Annual {ACM-SIAM} Symposium on Discrete
                  Algorithms, {SODA} 2013},
  pages        = {101--107},
  publisher    = {{SIAM}},
  year         = {2013},
  url          = {https://doi.org/10.1137/1.9781611973105.7},
  doi          = {10.1137/1.9781611973105.7},
  bibsource    = {dblp computer science bibliography, https://dblp.org}
}

@inproceedings{DBLP:conf/stoc/MahdianY11,
  author       = {Mohammad Mahdian and
                  Qiqi Yan},
  title        = {Online bipartite matching with random arrivals: an approach based
                  on strongly factor-revealing LPs},
  booktitle    = {Proceedings of the 43rd {ACM} Symposium on Theory of Computing, {STOC}
                  2011},
  pages        = {597--606},
  publisher    = {{ACM}},
  year         = {2011},
  url          = {https://doi.org/10.1145/1993636.1993716},
  doi          = {10.1145/1993636.1993716},
  bibsource    = {dblp computer science bibliography, https://dblp.org}
}

@inproceedings{KonradMM12,
  author       = {Christian Konrad and
                  Fr{\'{e}}d{\'{e}}ric Magniez and
                  Claire Mathieu},
  title        = {Maximum Matching in Semi-streaming with Few Passes},
  booktitle    = {Approximation, Randomization, and Combinatorial Optimization. Algorithms
                  and Techniques - 15th International Workshop, {APPROX} 2012, and 16th
                  International Workshop, {RANDOM} 2012. Proceedings},
  series       = {Lecture Notes in Computer Science},
  volume       = {7408},
  pages        = {231--242},
  publisher    = {Springer},
  year         = {2012},
  url          = {https://doi.org/10.1007/978-3-642-32512-0\_20},
  doi          = {10.1007/978-3-642-32512-0\_20},
  bibsource    = {dblp computer science bibliography, https://dblp.org}
}

@inproceedings{DBLP:conf/sosa/EdenFFS21,
  author       = {Alon Eden and
                  Michal Feldman and
                  Amos Fiat and
                  Kineret Segal},
  title        = {An Economics-Based Analysis of {RANKING} for Online Bipartite Matching},
  booktitle    = {4th Symposium on Simplicity in Algorithms, {SOSA} 2021},
  pages        = {107--110},
  publisher    = {{SIAM}},
  year         = {2021},
  url          = {https://doi.org/10.1137/1.9781611976496.12},
  doi          = {10.1137/1.9781611976496.12},
  bibsource    = {dblp computer science bibliography, https://dblp.org}
}

@inproceedings{DBLP:conf/soda/GoelM08,
  author       = {Gagan Goel and
                  Aranyak Mehta},
  title        = {Online budgeted matching in random input models with applications
                  to Adwords},
  booktitle    = {Proceedings of the Nineteenth Annual {ACM-SIAM} Symposium on Discrete
                  Algorithms, {SODA} 2008},
  pages        = {982--991},
  publisher    = {{SIAM}},
  year         = {2008},
  url          = {http://dl.acm.org/citation.cfm?id=1347082.1347189},
  bibsource    = {dblp computer science bibliography, https://dblp.org}
}

@article{DBLP:journals/sigact/BirnbaumM08,
  author       = {Benjamin E. Birnbaum and
                  Claire Mathieu},
  title        = {On-line bipartite matching made simple},
  journal      = {{SIGACT} News},
  volume       = {39},
  number       = {1},
  pages        = {80--87},
  year         = {2008},
  url          = {https://doi.org/10.1145/1360443.1360462},
  doi          = {10.1145/1360443.1360462},
  bibsource    = {dblp computer science bibliography, https://dblp.org}
}

@inproceedings{DBLP:conf/soda/GamlathKS19,
  author       = {Buddhima Gamlath and
                  Sagar Kale and
                  Ola Svensson},
  title        = {Beating Greedy for Stochastic Bipartite Matching},
  booktitle    = {Proceedings of the Thirtieth Annual {ACM-SIAM} Symposium on Discrete
                  Algorithms, {SODA} 2019},
  pages        = {2841--2854},
  publisher    = {{SIAM}},
  year         = {2019},
  url          = {https://doi.org/10.1137/1.9781611975482.176},
  doi          = {10.1137/1.9781611975482.176},
  bibsource    = {dblp computer science bibliography, https://dblp.org}
}

@book{thomas2006elements,
  title={Elements of information theory},
    author       = {Thomas M. Cover and
                  Joy A. Thomas},
  year={2006},
  publisher={Wiley-Interscience}
}

@inproceedings{DBLP:conf/soda/Derakhshan023,
  author       = {Mahsa Derakhshan and
                  Alireza Farhadi},
  editor       = {Nikhil Bansal and
                  Viswanath Nagarajan},
  title        = {Beating {(1} - 1/e)-Approximation for Weighted Stochastic Matching},
  booktitle    = {Proceedings of the 2023 {ACM-SIAM} Symposium on Discrete Algorithms,
                  {SODA} 2023},
  pages        = {1931--1961},
  publisher    = {{SIAM}},
  year         = {2023},
  url          = {https://doi.org/10.1137/1.9781611977554.ch74},
  doi          = {10.1137/1.9781611977554.CH74},
  bibsource    = {dblp computer science bibliography, https://dblp.org}
}

@article{huang2024online,
  title={Online matching: A brief survey},
  author={Huang, Zhiyi and Tang, Zhihao Gavin and Wajc, David},
  journal={arXiv preprint arXiv:2407.05381},
  year={2024}
}

@article{DBLP:journals/jacm/HuangKTWZZ20,
  author       = {Zhiyi Huang and
                  Ning Kang and
                  Zhihao Gavin Tang and
                  Xiaowei Wu and
                  Yuhao Zhang and
                  Xue Zhu},
  title        = {Fully Online Matching},
  journal      = {J. {ACM}},
  volume       = {67},
  number       = {3},
  pages        = {17:1--17:25},
  year         = {2020},
  url          = {https://doi.org/10.1145/3390890},
  doi          = {10.1145/3390890},
  bibsource    = {dblp computer science bibliography, https://dblp.org}
}

@inproceedings{DBLP:conf/soda/0002PTTWZ19,
  author       = {Zhiyi Huang and
                  Binghui Peng and
                  Zhihao Gavin Tang and
                  Runzhou Tao and
                  Xiaowei Wu and
                  Yuhao Zhang},
  editor       = {Timothy M. Chan},
  title        = {Tight Competitive Ratios of Classic Matching Algorithms in the Fully
                  Online Model},
  booktitle    = {Proceedings of the Thirtieth Annual {ACM-SIAM} Symposium on Discrete
                  Algorithms, {SODA} 2019},
  pages        = {2875--2886},
  publisher    = {{SIAM}},
  year         = {2019},
  url          = {https://doi.org/10.1137/1.9781611975482.178},
  doi          = {10.1137/1.9781611975482.178},
  bibsource    = {dblp computer science bibliography, https://dblp.org}
}

@inproceedings{DBLP:conf/focs/0002T0020,
  author       = {Zhiyi Huang and
                  Zhihao Gavin Tang and
                  Xiaowei Wu and
                  Yuhao Zhang},
  editor       = {Sandy Irani},
  title        = {Fully Online Matching {II:} Beating Ranking and Water-filling},
  booktitle    = {61st {IEEE} Annual Symposium on Foundations of Computer Science, {FOCS}
                  2020},
  pages        = {1380--1391},
  publisher    = {{IEEE}},
  year         = {2020},
  url          = {https://doi.org/10.1109/FOCS46700.2020.00130},
  doi          = {10.1109/FOCS46700.2020.00130},
  bibsource    = {dblp computer science bibliography, https://dblp.org}
}

@article{marshall1979inequalities,
  title={Inequalities: theory of majorization and its applications},
  author={Marshall, Albert W and Olkin, Ingram and Arnold, Barry C},
  year={1979},
  publisher={Springer}
}

@inproceedings{0.546ranking,
  author       = {Mahsa Derakhshan and Mohammad Roghani and Mohammad Saneian and Tao Yu},
  title        = {Improved Approximation for Ranking on General Graphs},
  booktitle    = {Proceedings of the 2026 {ACM-SIAM} Symposium on Discrete
                  Algorithms, {SODA} 2026},
  year         = {2026},
  publisher    = {{SIAM}},
  note         = {to appear}
}
\end{document}